\documentclass{amsart}
\usepackage{amsmath,amssymb}
\usepackage{graphicx,tikz,mathtools}
\newtheorem{theorem}{Theorem}[section]
\newtheorem{proposition}[theorem]{Proposition}

\newtheorem{lemma}[theorem]{Lemma}

\newcommand{\rank}{\operatorname{rank}}

\theoremstyle{definition}
\newtheorem{definition}[theorem]{Definition}

\theoremstyle{remark}
\newtheorem{remark}[theorem]{Remark}

\newcommand{\al}{\alpha}

\newcommand{\de}{\delta}
\newcommand{\ep}{\varepsilon}

\newcommand{\ka}{\kappa}

\newcommand{\vp}{\varphi}

\newcommand{\De}{\Delta}

\newcommand{\Si}{\Sigma}
\newcommand{\Om}{\Omega}

\newcommand{\cU}{{\mathcal U}}
\DeclareMathOperator\Real{Re}
\DeclareMathOperator\Imag{Im}
\def\CC{\mathbb{C}}

\def\RR{\mathbb{R}}

\newcommand{\cM}{{\mathcal M}}

\newcommand{\pd}{\partial}
\newcommand\minus\backslash
\newcommand{\id}{{\rm id}}

\newcommand\lan\langle
\newcommand\ran\rangle

\DeclareMathOperator\Div{div} 

\DeclareMathOperator\diag{diag}

\newcommand{\pdt}{\pd_{\mathrm{T}}\Om}
\newcommand{\pdb}{\pd_{\mathrm{B}}\Om}
\newcommand{\pdl}{\pd_{\mathrm{L}}\Om}

\renewcommand\leq\leqslant
\renewcommand\geq\geqslant
\newlength{\intwidth}

\newcommand\BOm{\overline\Om}

\numberwithin{equation}{section}

 \DeclareMathOperator\curl{curl}

\begin{document}

\title[Entanglement of vortices]{Entanglement of vortices\\ in the Ginzburg--Landau equations for superconductors}

 \author{Alberto Enciso}
 \address{Instituto de Ciencias Matem\'aticas, Consejo Superior de
   Investigaciones Cient\'\i ficas, 28049 Madrid, Spain}
 \email{aenciso@icmat.es}

 \author{Daniel Peralta-Salas}
 \address{Instituto de Ciencias Matem\'aticas, Consejo Superior de
   Investigaciones Cient\'\i ficas, 28049 Madrid, Spain}
 \email{dperalta@icmat.es}

%
%
\begin{abstract}
In 1988, Nelson proposed that neighboring vortex lines in high-temperature superconductors may become entangled with each other. In this article we construct solutions to the Ginzburg--Landau equations which indeed have this property, as they exhibit entangled vortex lines of arbitrary topological complexity.
\end{abstract}

\maketitle

\section{Introduction}

In the 1950s, Vitaly Ginzburg and Lev Landau developed a powerful phenomenological theory to provide a mathematical description of superconductors near the critical temperature, see e.g.~\cite{survey,Tink,Chap13}. The theory revolves around the so-called  {\em Ginzburg--Landau equations}\/, whose non-dimensional form is
\begin{align*}
\Big(\frac{1}{\kappa}\nabla -iA\Big)^2\Psi&=(|\Psi|^2-1)\Psi\,,\\
-\curl\curl A&=\frac{i}{2\kappa}(\bar\Psi\nabla\Psi-\Psi\nabla\bar\Psi)+|\Psi|^2A\,.
\end{align*}
The unknowns are the vector field $A:\Omega\to\RR^3$, which is the magnetic vector potential, and the complex-valued function $\Psi:\Omega\to\CC$, which is the order parameter representing the superconducting electron pairs. Here $\Omega\subset\RR^3$  is a bounded domain and $\kappa>0$ is the Ginzburg--Landau parameter, whose value determines whether the superconductor is of type I or II. The usual boundary conditions are
\begin{equation}\label{E.BC}
N\cdot (\nabla\Psi-i\Psi A)=0
\end{equation}
on $\pd\Om$.

A major problem in the study of high-temperature superconductors is the question of whether (and how) neighboring vortex lines in the glass or liquid phase may
entangle around each other~\cite{Berdi,NelsonN,RNature,RH04}, as was proposed by Nelson~\cite{Nelson} in 1988. Despite more than three decades of theoretical, numerical and experimental studies, the question of whether vortices can form an entangled state has not yet been convincingly answered. We recall that vortices are defined as the nodal lines
\[
Z_\Psi:=\{x\in\Om: \Psi(x)=0\}\,,
\]
which are the defects or singularities of the phase of the order parameter $\Psi$.

Our objective in this paper is to construct solutions to the Ginzburg--Landau system that indeed exhibit entangled vortices with complicated topologies. The solutions we shall construct have small amplitude, so it is convenient to write the Ginzburg--Landau equations as
\begin{align*}
(\ka^{-2}\Delta+1)\Psi&=2i\ka^{-1}A\cdot\nabla\Psi+|A|^2\Psi+|\Psi|^2\Psi+i\ka^{-1}\Div A\Psi\,,\\
\Delta A&=\nabla\Div A +\frac{i}{2\ka}(\bar\Psi\nabla\Psi-\Psi\nabla\bar\Psi)+|\Psi|^2A\,,
\end{align*}
where the Laplacian $\Delta$ acts on the vector field $A$ componentwise.

Given any smooth function $\chi:\Om\to\RR$, it is easy to check that the gauge transformation
\begin{equation*}
(A,\Psi)\mapsto (A+\nabla\chi,\; \Psi e^{i\chi})
\end{equation*}
takes solutions to solutions, so one can always pick a gauge where $A$ is divergence-free. However, for the construction later, it will be convenient not to fix such a gauge.

For concreteness, throughout the paper we will consider the Ginzburg--Landau equations on the cylindrical domain
$$\Om :=D_\rho\times (-\ell,\ell)\,,$$
as is usually done in applications. Here $\ell>0$ is the height of the cylinder and~$\rho$ is its radius; $D_\rho$ denotes the two-dimensional disk of radius $\rho$ centered at the origin.

Nelson's hypothesis concerns the existence of entangled (i.e., braided) vortices in high-temperature superconductors, whose precise definition is the following:

\begin{definition}
A {\em braid}\/ in the domain $\Om=D_\rho\times (-\ell,\ell)$ is a finite collection of pairwise disjoint smooth lines, diffeomorphic to $(-\ell,\ell)$, contained in $\Om$, whose endpoints are one at the bottom boundary $\pdb:=D_\rho\times\{-\ell\}$ and the other at the top boundary $\pdt:=D_\rho\times\{\ell\}$. These lines may be knotted and linked~\cite[Chapter~5.4]{Adams}.
\end{definition}

The following theorem is the main result of this article. We show that there exist solutions $(A,\Psi)$ to the Ginzburg--Landau equations on the cylinder~$\Om$, satisfying the natural boundary conditions on the lateral boundary $\pdl:= \pd D_\rho\times (-\ell,\ell)$, which exhibit a subset of isolated vortex lines isotopic to any prescribed braid $L$ in $\Om$. These vortex lines are structurally stable in the sense that any other complex-valued function that is close to $\Psi$ exhibits an isotopic subset of vortex lines.

\begin{theorem}\label{T:main}
Let $L$ be a braid in the cylinder $\Om $ and fix any positive integer~$r$ and any $\ep>0$. Then there exists a solution $(A,\Psi)$ to the Ginzburg--Landau equations in $\Om $ satisfying the boundary condition~\eqref{E.BC} on~$\pdl$,  such that $\Phi(L)$ is a subset of isolated vortex lines of $\Psi$.

Here  $\Phi $  is a smooth diffeomorphism of~$\BOm$, close to the identity in the sense that $\|\Phi-\id\|_{C^r(\Om)}<\ep$. Furthermore, these vortex lines are structurally stable. More precisely, take an open set $V\subset\Om$ which contains~$L$ and any~$\ep'>0$. Then there exists some $\de>0$ such that any function $\Psi_1$ satisfying $\|\Psi-\Psi_1\|_{C^r(V)}<\de$ has a subset of vortex lines given by $\Phi_1\circ\Phi(L)$, where $\Phi_1$ is a diffeomorphism of~$\BOm$ bounded as $\|\Phi_1-\id\|_{C^r(\Om)}<\ep'$.
\end{theorem}

Several remarks are in order. First, the nodal set $Z_\Psi$ may contain other components, but they are at a positive distance from $\Phi(L)$ because $\Phi(L)$ is an isolated nodal set. Second, note that there is no Neumann condition at the bottom and top boundaries of $\Om$, but this does not seem to be very important for the braid structure. Third, we emphasize that $(A,\Psi)$ is a small amplitude solution, so proving a similar result for large values of the order parameter function $\Psi$ (which forces the Ginzburg--Landau system to operate far from the linear regime) remains an interesting, and probably very hard, open problem.

The proof of Theorem~\ref{T:main}, which consists of five steps, is presented in Section~\ref{S.proof}. The proof of the auxiliary Proposition~\ref{L.psi}, which is more technical, is relegated to Section~\ref{S.psi}.

\section{Proof of the theorem}\label{S.proof}

The solution $(A,\Psi)$ we construct to the Ginzburg--Landau equations is a perturbation of a solution to the linearized equations, so the Ginzburg--Laudan system essentially operates in the monochromatic wave regime. The proof is divided in five steps. In the first one we construct a solution $\psi$ to the Helmholtz equation in $\Om$ which exhibits a subset of structurally stable vortex lines that are isotopic to the braid $L$. This solution is promoted to a solution of the nonlinear Ginzburg--Landau equations by using a fixed point argument in steps~2 to~4. To achieve this, it is crucial to solve certain elliptic PDEs in a cylindrical domain with suitable boundary conditions. The proof is completed in step~5 by using the structural stability of the nodal lines of $\psi$. All along this section, the integer~$r$ (which appears in the statement of Theorem~\ref{T:main}) is an arbitrary positive integer.

\subsection*{Step~1: A monochromatic wave with prescribed braided vortex lines}

In this first step we consider the extended cylinder $\Om':= D_\rho\times(-2\ell,2\ell)$ and a braid $L'\subset \Om'$ that is an ``extension'' of the braid $L$, i.e., $L'\cap \Om=L$. We prove that there is a monochromatic wave in $\Om'$ with suitable boundary conditions whose nodal set contains a braid that is isotopic to $L'$, the isotopy being as close to the identity as desired. In the proof we use techniques that we developed to study the level sets of harmonic functions in Euclidean space~\cite{EP13} and to tackle a conjecture of Berry on knotted nodal lines for eigenfunctions of some Schrodinger operators~\cite{jems18}.

The main new difficulty that arises is that we want to construct solutions in a bounded domain with Neumann boundary data, so the method of proof in~\cite{EP13,jems18} does not work directly and new technicalities have to be introduced. For future reference, for any $\ep>0$, we will say that two $C^r$ functions $f_1,f_2$ on~$\Om'$ are {\em $\ep$-close in the $C^r$-norm} if $\|f_1-f_2\|_{C^r(\Om')}<C\ep$ for some $\ep$-independent constant $C$.

\begin{proposition}\label{L.psi}
For any $\ep>0$, there exists a function $\psi\in C^\infty(\overline{\Om'},\mathbb C)$ and a diffeomorphism $\Phi_0:\overline{\Om'}\to \overline{\Om'}$ such that:
\begin{enumerate}
	\item $\psi$ satisfies the Helmholtz equation $(\ka^{-2}\De+1)\psi=0$ in $\Om'$ and the Neumann boundary condition $N\cdot \nabla\psi=0$ on the lateral boundary $\partial_L\Om':=\pd D_\rho\times(-2\ell,2\ell)$.
	\item $\Phi_0(L')$ is a subset of structurally stable vortex lines of $\psi$. More precisely, take an open set $V\subset\Om'$ such that~$V$ contains the braid $L'$. Then there exists some $\de>0$ such that any function $\psi_1\in C^r(V)$ satisfying $$\|\psi-\psi_1\|_{C^r(V)}<\de$$ has a subset of vortex lines~$S_1$ with $\Phi_1(S_1)= \Phi_0(L')$. Here $\Phi_1$ is a diffeomorphism of~$\overline V$ which is $\delta$-close to the identity in the $C^r$-norm.
	\item $\Phi_0$ is $\ep$-close to the identity in the $C^r$-norm on~$\Om'$.
\end{enumerate}
\end{proposition}

The proof of this result is presented in Section~\ref{S.psi}.

\subsection*{Step~2: Setting up a fixed point argument}

Let us take a contractible axisymmetric domain $\Om_1\subset \RR^3$ with $C^\infty$~boundary such that $\Om\subset \Om_1\subset \Om'$ (e.g., think of adding two caps to the domain $\Om$). Note, in particular, that $\pd\Om_1\supset\pdl$. It is standard (see e.g.~\cite{Uh76}) that, by slightly changing the boundary of this domain away from~$\pdl$ if necessary, we can safely assume that $\ka^2$ is not a Neumann eigenvalue of the Laplacian $-\Delta$ in the domain~$\Om_1$.

The linearization of the Ginzburg--Landau equation at the trivial solution $(A,\Psi):=(0,0)$ is
\begin{align*}
(\ka^{-2}\Delta+1)\dot\Psi&=0\,,\\
\Delta \dot A-\nabla \Div \dot A&=0\,,
\end{align*}
where $(\dot A,\dot\Psi)$ is our unknown. Therefore, our argument will be constructed around a fixed solution
\begin{align*}
(\dot A,\dot\Psi):=(a,\psi)
\end{align*}
of this linear system, where $a$ is any divergence-free vector field in $\RR^3$ such that $\Delta a=0$ and $N\cdot a = 0$ on~$\pdl'$ (e.g., $a(x):=(0,0,1)$ or $a(x):=(x_2,-x_1,0)$), and where $\psi\in C^\infty(\overline{\Om'},\mathbb C)$ is a monochromatic wave as in Proposition~\ref{L.psi}.

Specifically, in the next step we shall set an iteration of the form
\begin{equation}\label{E.linearized}
	\begin{aligned}
(\ka^{-2}\Delta+1)\Psi_{k+1}&=H(\Psi_{k},A_{k})\,,\\
 \Delta A_{k+1}&=F(\Psi_{k},A_{k})\,,
\end{aligned}
\end{equation}
where $(A_0,\Psi_0)$ is constructed using $(a,\psi)$ and where
\begin{align*}
H(\Psi ,A )&:=2i\ka^{-1}A \cdot\nabla\Psi +|A |^2\Psi +|\Psi |^2\Psi +i\ka^{-1}(\Div A) \, \Psi \,,\\
F(\Psi ,A )&:=\frac{i}{2\ka}(\bar\Psi \nabla\Psi -\Psi \nabla\bar\Psi )+|\Psi |^2A \,.
\end{align*}
For future reference, note that the H\"older norm of these functions can obviously be estimated using Young's inequality for products as
\begin{multline}
	\|H(\Psi ,A )\|_{C^{r-2,\al}}+ \|F(\Psi ,A )\|_{C^{r-2,\al}}\\
	 \leq C\big(\|\Psi\|_{C^{r-1,\al}}^2+ \|\Psi\|_{C^{r-1,\al}}^3+ \|A\|_{C^{r-1,\al}}^2+\|A\|_{C^{r-1,\al}}^3\big)\,,\label{E.cotaFH}
\end{multline}
where the constant $C$ depends on $r,\alpha$ and $\ka^{-1}$.

\subsection*{Step 3: Reduction to a modified boundary problem}

For $\eta>0$ small enough and each $k\geq0$, we consider the following functions defined on $\Om_1$:
\begin{align*}
\Psi_{k}&:=\frac{\eta}{2C_1} \psi+\widetilde\Psi_{k}\,,\\
A_{k}&:=\frac{\eta}{2C_2} a+\widetilde A_{k}\,.
\end{align*}
Here $(\widetilde A_k,\widetilde\Psi_k)$ will be solutions to the system~\eqref{E.linearized}, with  $(\widetilde A_0,\widetilde\Psi_0)=(0,0)$, and the constants are chosen as $C_1:=\|\psi\|_{C^{r,\alpha}(\Om_1)}$ and $C_2:=\|a\|_{C^{r,\alpha}(\Om_1)}$.

As we shall see in Lemma~\ref{L.BC} below, key ingredient in the iteration is the somewhat non-standard  choice of boundary conditions. Specifically, $\widetilde\Psi_{k+1}$ is the solution to the boundary value problem for the Helmholtz operator
\begin{equation}\label{E.Psik}
(\ka^{-2}\Delta+1)\widetilde\Psi_{k+1}=H(\Psi_{k},A_{k}) \quad \text{ in } \Om_1\,,
\end{equation}
with Neumann condition
\[
N\cdot \nabla\widetilde\Psi_{k+1}=i\Psi_kA_k\cdot N\quad \text{ on }\pd\Om_1\,.
\]
By the Fredholm alternative, this solution exists and is unique because we assumed~$\ka^2$ is not a Neumann eigenvalue of $-\Delta$ in $\Om_1$.

Likewise, $\widetilde A_{k+1}$ is the solution to the boundary problem:
\begin{equation}\label{eq.deltaA}
\Delta \widetilde A_{k+1}=F(\Psi_{k},A_{k}) \quad \text{ in } \Om_1
\end{equation}
with relative boundary conditions:
\[
\widetilde A_{k+1}\times N=0\quad \text{ and }\quad \Div \widetilde A_{k+1}=0\quad \text{ on } \pd\Om_1\,.
\]
Since $\Om_1$ is a contractible domain (and hence the space of  harmonic forms satisfying the relative boundary conditions is trivial), it is well known that $\widetilde A_{k+1}$ exists and is unique~\cite[Lemma~3.5.6]{Schwarz}.

It is clear that
\begin{enumerate}
\item The couple $(A_{k+1},\Psi_{k+1})$ solves the equations of the iteration in $\Om_1$.
\item $N\cdot(\nabla\Psi_{k+1}-i\Psi_kA_k)=0$ on $\pd_L\Om$ for all $k\geq0$.
\item $\Div A_{k+1}=0$ on $\partial \Om_1$ for all $k\geq0$.
\end{enumerate}
Assuming that the sequence $(A_k,\Psi_k)$ converges to a couple $(A,\Psi)$ in the $C^{r,\alpha}(\Om_1)$ norm, it follows from the definition of the iteration that $(A,\Psi)$ solves the equation
\begin{equation}\label{eqs}
\begin{aligned}
(\ka^{-2}\Delta+1)\Psi&=H(\Psi,A)\,,\\
 \Delta A&=F(\Psi,A)\,,
\end{aligned}
\end{equation}
in $\Om_1$, with boundary conditions $N\cdot(\nabla \Psi-i\Psi A)=0$ on $\partial_L\Om$ and $\Div A=0$ on $\partial\Om_1$. By elliptic regularity, $(A,\Psi)$ is $C^\infty$.

The following lemma shows that, in fact, the vector field $A$ is divergence-free in~$\Om_1$. To prove this we use a property of the vector field $F(\Psi,A)$ when $(A,\Psi)$ is a solution to Equations~\eqref{eqs}, and the boundary condition $\Div A=0$ on $\partial\Om_1$. This crucially uses our choice of relative boundary conditions to solve Equation~\eqref{eq.deltaA}.

\begin{lemma}\label{L.BC}
The solution $(A,\Psi)$ satisfies $\Div A=0$ in $\Om_1$, so in particular it solves the Ginzburg--Landau equations in $\Om_1$.
\end{lemma}
\begin{proof}
An easy computation shows that the couple $(A,\Psi)$ satisfies $\Div F(A,\Psi)=0$. Indeed,
\begin{align*}
&\Div F(A,\Psi)=\frac{i}{2\ka}(\bar\Psi\Delta\Psi-\Psi\Delta\bar\Psi)+A\cdot\nabla|\Psi|^2+|\Psi|^2\Div A\\
&=\frac{i}{2\ka}(2i\ka\bar\Psi A\cdot\nabla\Psi+2i\ka\Psi A\cdot\nabla\bar\Psi+2i\ka|\Psi|^2\Div A)+A\cdot\nabla|\Psi|^2+|\Psi|^2\Div A\\
&=0\,,
\end{align*}
where we have used the first equation that satisfies the solution $(A,\Psi)$ to pass to the second line. Accordingly, taking $\Div$ in the second equation that satisfies the couple $(A,\Psi)$, we get
\[
\Delta(\Div A)=0 \text{ in } \Om_1\,,\qquad \Div A=0 \text{ on } \partial\Om_1\,,
\]
 and hence $\Div A=0$ on $\Om_1$ as claimed. Finally, notice that this implies that the couple $(A,\Psi)$ satisfies the Ginzburg--Landau equations in the domain $\Om_1$.
\end{proof}

\subsection*{Step 4: Convergence of the scheme}

Armed with these estimates, one can now show the convergence of our iteration scheme by a standard induction argument. In particular, we show below that
\begin{align*}
\Psi_{k+1} = \frac{\eta}{2C_1}\psi +  O(\eta^2)\,,\\
A_{k+1}=\frac{\eta}{2C_2}a+ O(\eta^2)\,,
\end{align*}
where we use the notation $q= O(\eta^2)$ for terms bounded as $\|q\|_{C^{r,\alpha}(\Om_1)}\leq C\eta^2$.

To prove the convergence of the iteration, let us make the induction hypothesis that
$$\|A_{k'}\|_{C^{r,\alpha}(\Om_1)}<\eta\,,\qquad \|\Psi_{k'}\|_{C^{r,\alpha}(\Om_1)}<\eta$$
for all $k'\leq k$, where $\eta$ is a small enough constant. We shall next show that the same bounds hold for $A_{k+1}$ and $\Psi_{k+1}$.

Indeed, by definition
\[
\|A_{k+1}\|_{C^{r,\alpha}(\Om_1)}\leq \frac{\eta}{2}+\|\widetilde A_{k+1}\|_{C^{r,\alpha}(\Om_1)}\,,
\]
where $\widetilde A_{k+1}$ is the unique solution to the Hodge boundary problem~\eqref{eq.deltaA}. Standard Schauder estimates yield~\cite[Lemma~3.5.6]{Schwarz}
\[
\|\widetilde A_{k+1}\|_{C^{r,\alpha}(\Om_1)}\leq C\|F(\Psi_k,A_k)\|_{C^{r-2,\alpha}(\Om_1)}\,,
\]
where $C>0$ is a constant that depends on $r,\alpha,\Om_1$ but not on $k$. By the estimate~\eqref{E.cotaFH}, we then get
\begin{align*}
\|\widetilde A_{k+1}\|_{C^{r,\alpha}(\Om_1)}&\leq C\big(\|\Psi_k\|_{C^{r-1,\al}}^2+ \|\Psi_k\|_{C^{r-1,\al}}^3+ \|A_k\|_{C^{r-1,\al}}^2+\|A_k\|_{C^{r-1,\al}}^3\big)\\
&\leq C\eta^2\,.
\end{align*}
Therefore,
\[
\|A_{k+1}\|_{C^{r,\alpha}(\Om_1)}\leq \frac{\eta}{2}+C\eta^2<\eta\,,
\]
provided that $\eta$ is small enough.

Analogously, we can estimate
\[
\|\Psi_{k+1}\|_{C^{r,\alpha}(\Om_1)}\leq \frac{\eta}{2}+\|\widetilde\Psi_{k+1}\|_{C^{r,\alpha}(\Om_1)}\,,
\]
where $\widetilde\Psi_{k+1}$ is the unique solution to the boundary problem~\eqref{E.Psik}.
Again, by the Fredholm alternative, since $\ka^2$ is not a Neumann eigenvalue of $-\Delta$ in $\Om_1$, there exists a unique solution to this boundary problem, which by standard Schauder estimates is bounded as~\cite[Section~10.5]{Horman}
\[
\|\widetilde \Psi_{k+1}\|_{C^{r,\alpha}(\Om_1)}\leq C\big(\|H(\Psi_k,A_k)\|_{C^{r-2,\alpha}(\Om_1)}+\|\Psi_kA_k\|_{C^{r-1,\alpha}(\Om_1)}\big)\,,
\]
with $C$ a constant that depends on $r,\alpha,\Om_1$ but not on $k$. Together with~\eqref{E.cotaFH}, this results in
\begin{align*}
\|\widetilde \Psi_{k+1}\|_{C^{r,\alpha}(\Om_1)}&\leq C\big(\|\Psi_k\|_{C^{r-1,\al}}^2+ \|\Psi_k\|_{C^{r-1,\al}}^3+ \|A_k\|_{C^{r-1,\al}}^2+\|A_k\|_{C^{r-1,\al}}^3\big)\\
&\leq C\eta^2\,.
\end{align*}
Accordingly,
\[
\|\Psi_{k+1}\|_{C^{r,\alpha}(\Om_1)}\leq \frac{\eta}{2}+C\eta^2<\eta\,,
\]
provided that $\eta$ is small enough.

Since $(\widetilde A_0,\widetilde \Psi_0)=(0,0)$, we then infer that
\begin{equation}\label{eq.eta}
\|\Psi_{k}\|_{C^{r,\alpha}(\Om_1)}<\eta\,,\qquad \|A_{k}\|_{C^{r,\alpha}(\Om_1)}<\eta
\end{equation}
for all $k\geq 0$, and
\begin{equation}\label{eq.eta2}
\|\widetilde\Psi_{k}\|_{C^{r,\alpha}(\Om_1)}<C\eta^2\,,\qquad \|\widetilde A_{k}\|_{C^{r,\alpha}(\Om_1)}<C\eta^2
\end{equation}
for all $k\geq 0$ and some $k$-independent constant $C>0$ that does not depend on $\eta$.

Finally, to show that the sequence $(A,\Psi)$ is Cauchy, we notice that
$$
A_{k+1}-A_k=\widetilde A_{k+1}-\widetilde A_k\,,\qquad \Psi_{k+1}-\Psi_k=\widetilde\Psi_{k+1}-\widetilde\Psi_k\,,
$$
and these differences satisfy the boundary problems
\[
\Delta (\widetilde A_{k+1}-\widetilde A_k)=F(\Psi_{k},A_{k})-F(\Psi_{k-1},A_{k-1}) \text{ in } \Om_1
\]
with relative boundary conditions on $\pd\Om_1$, and
\[
(\ka^{-2}\Delta+1)(\widetilde\Psi_{k+1}-\widetilde\Psi_k)=H(\Psi_k,A_k)-H(\Psi_{k-1},A_{k-1}) \text{ in } \Om_1
\]
with Neumann boundary condition
$$
N\cdot \nabla (\widetilde \Psi_{k+1}-\widetilde\Psi_k)=iN\cdot(\Psi_kA_k-\Psi_{k-1}A_{k-1}) \quad \text{ on } \quad \partial\Om_{1}.
$$

For any $A_k,\Psi_k$ bounded as $\|A_k\|_{C^{r-1}(\Om_1)}+\|\Psi_k\|_{C^{r-1}(\Om_1)}<2\eta$,  thanks to the mean value theorem, one can argue as in the case of~\eqref{E.cotaFH} to conclude that
\begin{multline*}
	\|F(\Psi_{k+1},A_{k+1})-F(\Psi_{k},A_{k})\|_{C^{r-2,\alpha}(\Om_1)} + 	\|H(\Psi_{k+1},A_{k+1})-H(\Psi_{k},A_{k})\|_{C^{r-2,\alpha}(\Om_1)}\\
\leq C\eta\Big(\|A_{k}-A_{k-1}\|_{C^{r-1,\alpha}(\Om_1)}+\|\Psi_{k}-\Psi_{k-1}\|_{C^{r-1,\alpha}(\Om_1)}\Big)\,,
\end{multline*}
where the constant is independent of~$\eta<1$ and~$k$. Therefore, using the bound~\eqref{eq.eta} and Schauder estimates as before, we conclude that
\begin{multline*}
\|A_{k+1}-A_k\|_{C^{r,\alpha}(\Om_1)}+\|\Psi_{k+1}-\Psi_k\|_{C^{r,\alpha}(\Om_1)}\\
\leq C\eta \Big(\|A_{k}-A_{k-1}\|_{C^{r,\alpha}(\Om_1)}+\|\Psi_{k}-\Psi_{k-1}\|_{C^{r,\alpha}(\Om_1)}\Big)\,,
\end{multline*}
which yields convergence in the $C^{r,\alpha}$-norm provided that $\eta$ is chosen small enough so that $C\eta <\frac14$. This completes the proof that the sequence $(A_k,\Psi_k)$ converges to a couple $(A,\Psi)$ in $C^{r,\alpha}(\Om_1)$.

\subsection*{Step 5: Completion of the proof}
In the previous steps we have constructed a $C^\infty$ solution $(A,\Psi)$ to the Ginzburg--Landau equations in $\Om_1\supset\Om$ that satisfies the boundary condition
\[
N\cdot (\nabla\Psi-i\Psi A)=0
\]
on $\partial_L\Om$. Moreover, using the bound~\eqref{eq.eta2}, we conclude that it satisfies the estimate
\begin{equation*}
\|2C_1\eta^{-1}\Psi-\psi\|_{C^{r,\alpha}(\Om_1 )}+\|2C_2\eta^{-1}A-a\|_{C^{r,\alpha}(\Om_1)}\leq C\eta\,.
\end{equation*}
By construction, cf. Proposition~\ref{L.psi}, since $\Om_1\subset\Om'$, the function $\psi|_{\Om_1}$ has a subset of structurally stable vortex lines that is isotopic to $L'\cap\Om_1$, the diffeomorphism being $\varepsilon$-close to the identity in the $C^r$-norm. In particular, $\psi|_{\Om}$ has a subset of structurally stable vortex lines isotopic to $L=L'\cap\Om$. This fact and the previous estimate imply that $2C_1\eta^{-1}\Psi|_{\Om}$, and hence $\Psi|_{\Om}$, has a subset of vortex lines isotopic to $L$. Specifically, there is a smooth diffeomorphism $\Phi:\overline\Om\to\overline\Om$ which is $(\varepsilon+\eta)$-close to the identity in the $C^r$-norm, and such that $\Phi(L)$ is a subset of vortex lines of $\Psi|_{\Om}$. Since $\varepsilon+\eta$ can be taken as small as desired, this completes the proof of the theorem.
\begin{remark}
The solution $A$ is a perturbation of the vector field $a$ in $\Om$. Since $a$ is fairly explicit, this provides a good understanding of the orbits of $A$ in the perturbative regime, but this does not seem to be relevant for applications.
\end{remark}

\section{Proof of Proposition~\ref{L.psi}}\label{S.psi}

It is convenient to take a larger cylindrical domain $\widetilde\Om=D_\rho\times (-\widetilde\ell,\widetilde\ell)$, with $\widetilde\ell\in (2\ell,3\ell)$, and a braid $\widetilde L\subset\widetilde\Om$ which is an ``extension'' of the braid $L'$, i.e.,
$$\widetilde L\cap\Om'=L'\,.$$
An easy application of Whitney's approximation theorem ensures that, by perturbing the braid $\widetilde L$ slightly if necessary (see e.g.~\cite[Section 2.5]{Hirsch}), we can assume that it is a real analytic submanifold of $\widetilde\Om$. Let us denote by $\{\widetilde L_k\}_{k=1}^n$ the connected components of $\widetilde L$. Each component
$\widetilde L_k$ is an analytic open curve without self-intersections, and we claim that we can write the curve $\widetilde L_k$ as the transverse intersection
of two surfaces $\Si_k^1$ and $\Si_k^2$.

Indeed, each curve $\widetilde L_k$ is contractible, so it has
trivial normal bundle. Then there exists an analytic
submersion $\Theta_k:W_k\to\RR^2$, where $W_k\subset\widetilde\Om$ is a tubular
neighborhood of $\widetilde L_k$ and $\Theta_k^{-1}(0)=\widetilde L_k$. We can then take the
analytic surfaces
$$\Si^1_k:=\Theta_k^{-1}((-1,1)\times\{0\})\subset W_k\,, \qquad \Si^2_k:=\Theta_k^{-1}(\{0\}\times(-1,1))\subset W_k\,.$$
Since $\Theta_k$ is a submersion, these surfaces intersect transversally at
$$\widetilde L_k=\Si^1_k\cap\Si^2_k\,.$$

Now that we have expressed the component $\widetilde L_k$ as the intersection
of two real analytic surfaces $\Si^1_{k}$ and $\Si^2_k$, we can
consider the following Cauchy problems, with $m=1,2$:
\begin{equation*}
\kappa^{-2}\Delta u_k^m+u_k^m=0\,,\qquad u_k^m|_{\Si^m_k}=0\,,\qquad \pd_N u_k^m|_{\Si^m_k}=1\,.
\end{equation*}
Here $\pd_N$ denotes a normal derivative at the corresponding
surface. The Cauchy--Kovalevskaya theorem then grants the existence of
solutions $u_k^m$ to this Cauchy problem in the closure of small
neighborhoods $U^m_{k}\subset\widetilde\Om$ of each surface $\Si^m_{k}$. We can safely
assume that the tubular neighborhoods $U^1_k\cap U^2_k$ are small
enough so that the neighborhoods corresponding to distinct components
are disjoint. Now we take the union of these pairwise disjoint tubular neighborhoods,
\[
U:=\bigcup_{k}(U^1_k\cap U^2_k)\subset\widetilde\Om\,,
\]
and define a
complex-valued function $\varphi$ on the set $U$ as
$$\varphi|_{U^1_{k}\cap U^2_k}:=u_k^1+iu_k^2\,.$$

The following properties of $\varphi$ are clear from the construction:
\begin{enumerate}
\item $\varphi$ satisfies the equation
$$\kappa^{-2}\Delta\varphi+\varphi=0$$
in the tubular neighborhood $U$ of the braid $\widetilde L$.
\item $U$ can be taken small enough so that the nodal set of $\varphi$ is precisely $\widetilde L$, i.e., $\widetilde L=\varphi^{-1}(0)$.
\item The intersection of the zero sets of the real and imaginary parts of $\varphi$ on $\widetilde L$ is transverse, i.e.,
\begin{equation}\label{trans}
\rank(\nabla \Real \varphi(x),\nabla \Imag \varphi(x))=2
\end{equation}
for all $x\in \widetilde L$.
\end{enumerate}

Denote by $S$ a compact subset of $U$ whose interior
contains $\overline {L'}$ and $\widetilde\Om\backslash S$ is connected. Our next goal is to construct a solution of the
Helmholtz equation in~$\Om'$ that approximates the local solution $\varphi$ in the set $S\cap \Om'$. To this end, let us take a smooth function $\chi:\RR^3\to\RR$ equal to $1$ in a neighborhood of $S$ and identically zero outside $U\cap\Big(D_\rho\times(-\widetilde\ell+\delta_0,\widetilde\ell -\delta_0)\Big)$, with $\delta_0>0$ a small but fixed constant so that $\widetilde\ell-\delta_0>2\ell$ and $S\subset D_\rho\times(-\widetilde\ell+\delta_0,\widetilde\ell -\delta_0)$. We then define
a smooth extension~$\vp_0$ of the function $\varphi$ to $\widetilde\Om$ by
setting
$$\vp_0:=\chi \varphi\,.$$
By construction, this function is compactly supported in $\widetilde\Om$ and $\varphi_0=\varphi$ in a neighborhood of $S$.

Let us now write the function $\vp_0$ as an integral involving the Neumann Green's function of the domain $\widetilde\Om$. We first observe that, by the monotonicity principle, for any fixed $\ka>0$ there is a constant $\widetilde\ell\in (2\ell,3\ell)$ such that $\ka^2$ is not a Neumann eigenvalue of $-\Delta$ in $\widetilde\Om$. In what follows we fix such a constant $\widetilde\ell$ (which depends on $\ka$). The next lemma on the existence of a Neumann Green's function is elementary, but we provide a proof for the sake of completeness:

\begin{lemma}\label{L.Green}
Let $\diag$ denote the diagonal of $\widetilde\Om\times\widetilde\Om$. For some $\widetilde\ell\in (2\ell,3\ell)$, there exists a distribution $G\in C^\infty(\widetilde\Om\times\widetilde\Om)\backslash \diag$ with the following properties:
\begin{enumerate}
\item $G$ is symmetric: $G(x,y)=G(y,x)$.
	\item $(\ka^{-2}\De_x +1) G( x,y)=\de(x-y)$.
	\item $N(x)\cdot \nabla_x G(x,\cdot)=0$ for all $x\in \pd\widetilde\Om$.
\item $|G(x,y)|\leq \frac{C}{|x-y|}$ for some constant $C>0$ and any $(x,y)\in (\widetilde\Om\times\widetilde\Om)\backslash \diag$
\end{enumerate}
\end{lemma}
\begin{proof}
Since $\ka^2$ is not a Neumann eigenvalue of $-\Delta$ in $\widetilde\Om$, it follows from Fredholm's alternative that there exists a unique Green's function $G\in C^\infty(\widetilde\Om\times\widetilde\Om)\backslash \diag$ satisfying
\[
(\ka^{-2}\De_x +1) G( x,y)=\de(x-y)
\]
in~$\widetilde\Om$ and the Neumann condition $N(x)\cdot \nabla_x G(x,\cdot)=0$ for all $x\in \pd\widetilde\Om$. We claim that $G$ is symmetric. Indeed, using Green's second identity it follows that
\begin{align*}
G(z,y)-G(y,z)&=\int_{\pd\widetilde\Om}\big(G(x,y)\nabla_x G(x,z)\cdot N(x)-G(x,z)\nabla_x G(x,y)\cdot N(x)\big)d\sigma(x)\\
&=0\,,
\end{align*}
where we have used that $G$ satisfies the Neumann condition on $\partial\widetilde\Om$. Finally, the estimate in the item (iv) follows from a general property of singularities of solutions to second order elliptics PDEs, cf.~\cite{GT}. This completes the proof of the lemma.
\end{proof}

Setting $\rho:= \ka^{-2}\Delta \varphi_0+\varphi_0$, it follows from the fact that $\ka^2$ is not a Neumann eigenvalue of $-\Delta$ in $\widetilde\Om$ that
\begin{equation}\label{intbv}
\vp_0(x)=\int_{\widetilde\Om} G(x,y)\, \rho(y)\, dy\,.
\end{equation}
The compact support of the complex-valued function $\rho$ is contained in the open set~$\widetilde\Om\backslash S$, and its distance to the caps $D_\rho\times\{-\widetilde\ell\}$ and $D_\rho\times \{\widetilde \ell\}$ is at least $\delta_0>0$. Next, observe that an easy continuity argument ensures that one
can approximate the integral~\eqref{intbv} uniformly in the compact set~$S$ by a finite
Riemann sum of the form
\begin{equation}\label{vp1}
\vp_1(x):=\sum_{j=-J}^J \rho_j\, G(x,x_j)\,.
\end{equation}
Specifically, for any $\de>0$ there is a large integer $J$, complex
numbers $\rho_j$ and points $x_j\in \text{supp } \rho\subset \widetilde\Om\backslash S$ such that the
finite sum~\eqref{vp1} satisfies
\begin{equation}\label{est1}
\|\varphi_1-\varphi_0\|_{C^0(S)}<\de\,.
\end{equation}

In the following lemma we show how to ``sweep'' the singularities of
the function~$\varphi_1$ in order to approximate it in the set $S$ by another function $\varphi_2$ whose singularities are contained in $D_\rho\times (\widetilde\ell-\frac{\delta_0}{2},\widetilde\ell)$. The proof is based on a duality argument and the
Hahn--Banach theorem, and is an adaptation to our setting of the proof of~\cite[Lemma 4.1]{jems18}.

\begin{lemma}\label{L.approx}
For any $\de>0$, there is a finite set of points
$\{z_j\}_{j=-J'}^{J'}$ in $D_\rho\times (\widetilde\ell-\frac{\delta_0}{2},\widetilde\ell)\subset\widetilde\Om$ and complex numbers $c_j$ such that the
finite linear combination
\begin{equation}\label{eqmwx}
\psi(x):=\sum_{j=-J'}^{J'} c_j\, G(x,z_j)
\end{equation}
approximates the function $\varphi_1$ uniformly in $S$:
\begin{equation}\label{GtG}
\|\psi-\varphi_1\|_{C^0(S)}<\de\,.
\end{equation}
\end{lemma}
\begin{proof}
Consider the space $\cU$ of all complex-valued functions on $\widetilde\Om$ that are finite linear
combinations of the form~\eqref{eqmwx}, where the points
$z_j$ range over $D_\rho\times (\widetilde\ell-\frac{\delta_0}{2},\widetilde\ell)$ and where the
constants $c_j$ may take arbitrary complex values. Restricting these functions to the set $S$, $\cU$ can
be regarded as a subspace of the Banach space $C^0(S)$ of continuous
complex-valued functions on $S$.

By the Riesz--Markov theorem, the dual of $C^0(S)$ is  the space
$\cM(S)$ of the finite complex-valued Borel measures on $\widetilde\Om$ whose support
is contained in the set~$S$. Let us take any measure
$\mu\in\cM(S)$ such that $\int_{\widetilde\Om} f\, d\mu=0$ for all
$f\in \cU$. Using that the Green's function $G$ is symmetric, cf. item (i) in Lemma~\ref{L.Green}, we now define a complex-valued function $F$ as
\[
F(x):=\int_{\widetilde\Om} G(\widetilde x, x)\,d\mu(\widetilde x)\,,
\]
which is in $L^1(\widetilde\Om)$ by the estimate (iv) in Lemma~\ref{L.Green}.
Furthermore, it is clear that $F$ satisfies the equation
\[
\ka^{-2}\De F+F=\mu
\]
distributionally in~$\widetilde\Om$: for any $\phi\in C^2_c(\widetilde\Om)$, one has
\[
\int_{\widetilde\Om} F\, (\ka^{-2}\De+1)\phi\, dx =\int_{\widetilde\Om} \phi\, d\mu\,.
\]

Notice that $F$ is identically zero on $D_\rho\times (\widetilde\ell-\frac{\delta_0}{2},\widetilde\ell)$ by the definition of the
measure~$\mu$ and that $F$
satisfies the elliptic
equation
\[
\De F+F=0
\]
in $\widetilde\Om\minus S$, so $F$ is analytic in this set. Hence, since
$\widetilde\Om\minus S$ is connected and contains the set $D_\rho\times (\widetilde\ell-\frac{\delta_0}{2},\widetilde\ell)$, by
analyticity the function $F$ must vanish on the complement of~$S$. It
then follows that the measure $\mu$ also annihilates any complex-valued function of the
form $G(x,x_j)$
because, as the points $x_j$ do not belong to $S$,
\[
0=F(x_j)=\int_{\widetilde\Om} G(x,x_j)\,d\mu(x)\,.
\]
Therefore,
\[
\int_{\RR^3}\varphi_1\,d\mu= \sum_{j=-J}^J \rho_j \int_{\widetilde\Om} G(x,x_j)\,d\mu(x)=0\,,
\]
which implies that $\varphi_1$ can be
uniformly approximated on~$S$ by elements of the subspace $\cU$ as a
consequence of the Hahn--Banach theorem. This completes the proof of the lemma.
\end{proof}

To complete the proof of Proposition~\ref{L.psi}, we observe that
\begin{equation}\label{eqwmi}
\ka^{-2}\De \psi+\psi=0
\end{equation}
in the set $D_\rho\times (-\widetilde\ell+\delta_0,\widetilde\ell-\delta_0)$, which contains $S$ and $\Om'$. Moreover, it satisfies the Neumann boundary condition $N\cdot \nabla\psi=0$ on the lateral boundary $\pd D_\rho\times (-\widetilde\ell+\delta_0,\widetilde\ell-\delta_0)\supset \partial_L\Om'$, and the estimate
\[
\|\psi-\vp\|_{C^0(S)}<2\delta\,,
\]
where we have used the bounds~\eqref{est1} and~\eqref{GtG} and that $\vp_0=\vp$ on a neighborhood of $S$. In fact, since $\vp$ also satisfies the Helmholtz equation in a
neighborhood of
the compact set $S$, standard elliptic estimates imply that
the uniform estimate can be promoted to the $C^r$
bound
\begin{equation}\label{lastbound}
\|\psi-\varphi\|_{C^r(S)}<C_r\de\,.
\end{equation}

Finally, taking into account that the braid $\widetilde L$ is exactly the nodal
set of $\varphi$ and satisfies the transversality
condition~\eqref{trans}, and using the estimate~\eqref{lastbound}, the fact that~$\overline{L'}$ is contained in $S$ and Thom's isotopy theorem~\cite[Theorem 20.2]{AR}, we conclude that there is a diffeomorphism $\Phi_0$ of $\overline{\Om'}$ such that
$\Phi_0(L')$ is a union of components of the zero set
$\psi^{-1}(0)\cap\Om'$. Moreover, the diffeomorphism $\Phi_0$ is $C^r$-close to
the identity. The structural stability of the link $\Phi_0(L')$ for the function
$\psi|_{\Om'}$ also follows from Thom's isotopy theorem and the fact that
$\psi|_{\Om'}$ satisfies the transversality condition
$$
\rank(\nabla \Real \psi(x),\nabla \Imag \psi(x))=2
$$
for all $x\in\Phi_0(L')$. This last equation is a consequence of the
$C^r$-estimate~\eqref{lastbound}, the fact that the function~$\vp$
satisfies the transversality estimate~\eqref{trans} by definition, and
the fact that transversality is an open property under $C^1$-small
perturbations. This completes the proof of the proposition.

\section*{Acknowledgements}

The authors are very grateful to Lourdes F\'abrega for introducing to them the problem of vortex entanglement in high-temperature superconductors and indicating Ref.~\cite{survey}. This work has received funding from the European Research Council (ERC) under the European Union's Horizon 2020 research and innovation programme through the grant agreement~862342 (A.E.). It is partially supported by the grants CEX2019-000904-S, RED2022-134301-T and PID2022-136795NB-I00 (A.E. and D.P.-S.) funded by MCIN/AEI/10.13039/501100011033, and Ayudas Fundaci\'on BBVA a Proyectos de Investigaci\'on Cient\'ifica 2021 (D.P.-S.).

%
%
%

\bibliographystyle{amsplain}

\begin{thebibliography}{99}\frenchspacing

\bibitem{AR}
R. Abraham, J. Robbin, {\em Transversal mappings and flows}, Benjamin, New York, 1967.

\bibitem{Adams}
C.C. Adams, {\em The knot book}, AMS, Providence, 2004.

\bibitem{Berdi}
G.R. Berdiyorov, M.V. Milosevic, F. Kusmartsev, et al., Josephson vortex loops in nanostructured Josephson junctions, Sci. Rep. 8 (2018) 2733.

\bibitem{survey}
G. Blatter, M.V. Feigelman, V.B. Geshkenbein, A.I. Larkin, V.M. Vinokur, Vortices in high-temperature superconductors, Rev. Mod. Phys. 66 (1994) 1125--1388.

\bibitem{Chap13}
J. Chapman, An introduction to Ginzburg Landau vortices, Lecture notes from a tutorial course at the School ``Around vortices: from continuum to quantum mechanics'', IMPA, March 12-21, 2014.

\bibitem{EP13}
A. Enciso, D. Peralta-Salas, Submanifolds that are level sets of solutions to a second-order elliptic PDE, Adv. Math. 249 (2013) 204--249.

\bibitem{jems18}
A. Enciso, D. Hartley, D. Peralta-Salas, A problem of Berry and knotted zeros in the eigenfunctions of the harmonic oscillator, J. Eur. Math. Soc. 20 (2018) 301--314.

\bibitem{GT}
D. Gilbarg, J. Serrin, On isolated singularities of solutions of second order elliptic differential equations, J. Anal. Math. 4 (1954) 309--340.

\bibitem{Hirsch}
M.W. Hirsch, {\em Differential topology}, Springer, New York, 1976.

\bibitem{Horman}
L. H\"ormander, \emph{Linear partial differential operators}, Springer, Berlin, 1976.

\bibitem{Nelson}
D.R. Nelson, Vortex entanglement in high-$T_c$ superconductors, Phys. Rev. Lett. 60 (1988) 1973--1976.

\bibitem{NelsonN}
D.R. Nelson, Vortices weave a tangled web, Nature 430 (2004) 839--840.

\bibitem{RNature}
C. Reichhardt, Vortices wiggled and dragged, Nature Phys. 5 (2009) 15--16.

\bibitem{RH04}
C. Reichhardt, M.B. Hastings, Do vortices entangle? Phys. Rev. Lett. 92 (2004) 157002.


\bibitem{Schwarz}
G. Schwarz, \emph{Hodge decomposition, a method for solving boundary value problems}, Springer, Berlin, 1995.

\bibitem{Tink}
M. Tinkham, \emph{Introduction to superconductivity}, McGraw-Hill, New York, 1996.

\bibitem{Uh76}
K. Uhlenbeck, Generic properties of eigenfunctions, Amer. J. Math. 98 (1976) 1059--1078.

\end{thebibliography}

\end{document}